\documentclass[10pt,twocolumn]{article}
\usepackage[T1]{fontenc}
\usepackage[utf8]{inputenc}
\usepackage{lmodern}
\usepackage[a4paper,margin=18mm,columnsep=8mm]{geometry}
\usepackage{amsmath,amssymb,amsthm,mathtools}
\usepackage{booktabs}
\usepackage{array}
\usepackage{microtype}
\usepackage{enumitem}
\usepackage{natbib}
\usepackage[hidelinks]{hyperref}
\usepackage{url}
\usepackage{xcolor}
\usepackage{orcidlink}

\newtheorem{theorem}{Theorem}[section]
\newtheorem{proposition}[theorem]{Proposition}
\newtheorem{lemma}[theorem]{Lemma}
\theoremstyle{definition}
\newtheorem{definition}[theorem]{Definition}

\title{Single-Machine Scheduling with Interval Predictions and Costly Preemption}
\author{Lachlan Bridges\orcidlink{0000-0003-0936-4372}\\
\small Independent Researcher, Adelaide, South Australia, Australia\\
\small Corresponding author: \href{mailto:lachlanjbridges@gmail.com}{lachlanjbridges@gmail.com}}
\date{}

\begin{document}
\maketitle

\begin{abstract}
We study the single-machine total-completion-time problem $1||\sum C_j$ when processing times are unknown but each job comes with a reported interval. Jobs are initially processed in nondecreasing order of reported upper bound. If a job is still unfinished after receiving that much service, the reported upper bound has been violated and the policy switches to a resumable geometric fallback. Each interruption of an unfinished job incurs an additive penalty $\kappa$.

For a residual set of $m$ jobs, known size ratio $D=s_1/s_0$, normalized interruption penalty $\lambda=\kappa/s_0$, and geometric depth $K$, we derive an explicit worst-case coefficient $\Psi_{m,D,\lambda}(K)$. When $\lambda>0$, a finite minimizing depth exists and can be chosen after the trigger from the observed number of unfinished jobs; the optimal depth decreases with the interruption penalty and increases with the residual set size. We also derive the exact worst-case coefficient $R_{n,D}$ for an arbitrary nonpreemptive list when all processing times lie in a known bounded range. These bounds give guarantees for valid intervals, a single interval failure, arbitrary reports, no-trigger outcomes, and random instances. In the single-failure regime, they also give a precise condition under which the proved fallback guarantee is smaller than the bounded-range continuation guarantee.
\end{abstract}


\section{Introduction and contributions}
The underlying clairvoyant problem is the classical single-machine objective $1||\sum C_j$. We assume instead that each job arrives with a reported interval $[\ell_j,u_j]$ and that interrupting an unfinished job carries a cost. The scheduler first processes jobs in nondecreasing order of $u_j$. If a job is still unfinished after receiving $u_j$ units of service, the scheduler has learned something concrete: that reported upper bound was wrong. At that point it may continue the original list or use the new information to change how the unfinished jobs are scheduled.

This information pattern is useful when completed service should be preserved, so restarting a job would waste work, but frequent switching is itself disruptive. The resulting decision is not simply whether to preempt; it is how aggressively to interleave the remaining jobs after the original report has failed. Our fallback uses geometric cumulative service thresholds and keeps all completed service.

We measure interruption cost in the same units as total completion time. Thus $\kappa\ge0$ is the additive penalty for switching away from an unfinished job, and $\lambda=\kappa/s_0$ is the corresponding dimensionless parameter. Lower endpoints and interval widths are used in the performance guarantees, while upper bounds determine the initial order and the observable trigger.

\paragraph{Contributions.}
\begin{enumerate}[leftmargin=*]
\item We study a resumable single-machine model in which violation of a reported upper bound triggers a change of scheduling rule, and interruptions of unfinished jobs carry an explicit additive cost.
\item For every finite geometric depth $K$, we prove a joint bound on residual total completion time and interruption cost.
\item When the interruption penalty is positive, the number of unfinished jobs observed at the trigger determines a finite bound-minimizing depth, with
\[
\lambda\uparrow\Longrightarrow K^*\downarrow,\qquad m\uparrow\Longrightarrow K^*\uparrow.
\]
\item We derive the exact worst-case ratio $R_{n,D}$ of an arbitrary nonpreemptive list to SPT when all processing times lie in $[s_0,s_1]$, and identify sharp step-vector examples.
\item Combining the pre-trigger list bound with the fallback bound gives full-policy guarantees and a precise condition under which, in the single-failure regime, the fallback guarantee is smaller than the bounded-range list bound.
\end{enumerate}
The mathematical contribution is the interaction of these ingredients. The upper-bound violation is observed during ordinary execution rather than through a separate information-acquisition action; completed service is retained; interruptions carry an explicit additive cost; and, after the trigger, the fallback depth is chosen from the observed number of unfinished jobs. The exact bounded-range list coefficient provides a sharp comparator for the resulting policy guarantees.

\section{Model, interruption cost, and benchmark identities}
\begin{definition}[Instance and reports]
There are $n\ge1$ jobs, all released at time zero, on one unit-speed resumable machine. Job $j$ has processing requirement $p_j>0$ and report $I_j=[\ell_j,u_j]$ with $0<\ell_j\le u_j<\infty$. A deterministic label rule breaks ties. The joint-coverage event is
\[
V_0=\{\ell_j\le p_j\le u_j\ \forall j\}.
\]
Whenever an interval-width guarantee is invoked, $\rho\ge1$ satisfies $u_j/\ell_j\le\rho$. For the fallback results we assume
\[
0<s_0\le p_j\le s_1<\infty,\qquad D=s_1/s_0.
\]
The bounds $s_0,s_1$ (and hence $D$) are known to the scheduler whenever those results are used. Here $s_0$ is a global lower bound on the true processing times; it is distinct from the reported lower endpoints $\ell_j$, which belong to the possibly erroneous job reports.
\end{definition}

\begin{definition}[Schedules and interruption cost]
A resumable schedule is a selector $\sigma(t)\in\{0,1,\ldots,n\}$, with $0$ denoting idle. It is admissible if every job completes and the selector has a locally finite sequence of positive-length maximal constant service intervals through final completion. An interruption is charged exactly when the scheduler switches away from a job that is still unfinished; starts from idle and departures at completion are free, and adjacent same-job phases merge. Let
\[
\begin{aligned}
F_A(p)&=\sum_jC_j^A(p),\\
P_A(p)&=\#\{\text{charged interruptions}\}.
\end{aligned}
\]
\[
J_A(p)=F_A(p)+\kappa P_A(p),\qquad \lambda=\kappa/s_0.
\]
The parameter $\kappa$ is known to the scheduler. Appendix~\ref{app:fractional} records a fractional comparison that uses total completion time only; the physical policy studied in the main text is always resumable and single-server.
\end{definition}

\begin{lemma}[Clairvoyant benchmark]
If $x_{(1)}\le\cdots\le x_{(m)}$ are positive sizes, SPT minimizes total completion time even among preemptive schedules, and
\[
\operatorname{OPT}(x)=\sum_{i=1}^m(m-i+1)x_{(i)}=\sum_i x_i+\sum_{i<j}\min\{x_i,x_j\}.
\]
Because SPT has no charged interruption, it also minimizes $J$ among admissible resumable schedules for every $\kappa\ge0$.
\end{lemma}
\begin{proof}
Order jobs of any feasible schedule by completion time and process them nonpreemptively in that order. Each corresponding prefix completes no later; adjacent interchange then gives SPT. Expanding the SPT sum pairwise gives the identity.
\end{proof}

\begin{proposition}[Sharp bounded-range list bound]\label{prop:range}
Put $S_0=0$ and $S_t=t(t+1)/2$ for $t\ge1$, and define
\[
R_{n,D}:=\max_{0\le k\le n}\frac{S_{n-k}+D(S_n-S_{n-k})}{(S_n-S_k)+DS_k}.
\]
Every nonpreemptive list $L$ satisfies
\[
F_L(p)\le R_{n,D}\operatorname{OPT}(p),
\]
and the factor is sharp. Equality is attained for any maximizing $k$ by $n-k$ jobs of size $s_0$ and $k$ jobs of size $s_1$, scheduled in LPT order. In particular $R_{n,1}=1$. On joint coverage, the upper-endpoint list also satisfies the $\rho$ bound, so define
\[
A_{n,D,\rho}:=\min\{\rho,R_{n,D}\},\qquad F_U(p)\le A_{n,D,\rho}\operatorname{OPT}(p).
\]
\end{proposition}
\begin{proof}
The case $D=1$ is immediate because all processing times are equal, so assume $D>1$. Normalize $s_0=1$ and sort actual sizes as $1\le x_1\le\cdots\le x_n\le D$. For a fixed multiset, the rearrangement inequality makes LPT the worst list, with numerator $\sum_i i x_i$, while SPT has denominator $\sum_i(n-i+1)x_i$.

The ratio has positive denominator. If $x=\sum_r\gamma_r v^{(r)}$ is a convex combination of vertices of the ordered range polytope, then its ratio is a convex combination of the vertex ratios with weights proportional to $\gamma_r\sum_i(n-i+1)v_i^{(r)}$; hence a maximum is attained at a vertex. After the affine change $y_i=(x_i-1)/(D-1)$, the vertices are exactly the monotone $0$--$1$ step vectors: any interior fractional block can be perturbed up and down while preserving monotonicity. Thus $x_1=\cdots=x_{n-k}=1$ and $x_{n-k+1}=\cdots=x_n=D$ for some $k$. Evaluating LPT and SPT gives exactly the displayed ratio.

For the covered upper-endpoint bound, if job $i$ precedes job $j$ then $u_i\le u_j$, and coverage plus $u_j/\ell_j\le\rho$ gives
\[
p_i\le u_i\le u_j\le\rho\ell_j\le\rho p_j.
\]
Hence the pairwise contribution $p_i$ of this ordered pair is at most $\rho\min\{p_i,p_j\}$; the singleton terms satisfy $p_i\le\rho p_i$. Summing the pairwise decomposition from Lemma~2.3 yields $F_U(p)\le\rho\operatorname{OPT}(p)$. Taking the minimum with the independent bounded-range bound proves the corollary.
\end{proof}

\begin{definition}[Coverage events]
Let $V_1$ be the event that exactly the job producing the first upper-bound trigger violates its interval while all other intervals cover, and let $V_c=(V_0\cup V_1)^c$. On a triggering outcome let $M$ be the number of unfinished jobs at trigger time; on a no-trigger outcome set $M=0$.
\end{definition}

\section{A fractional comparison}
For comparison with the physical policy, Appendix~\ref{app:fractional} records a fractional policy that combines the upper-endpoint list with equal sharing. It satisfies a coverage-dependent bound and a report-independent bound based on the sharp equal-sharing factor $2n/(n+1)$. This comparison is not used in the fallback analysis below.

\section{Upper-bound violation and geometric fallback}
\begin{definition}[Upper-bound trigger]
Order jobs by nondecreasing $u_j$ and run that list nonpreemptively. When a running job has received exactly $u_j$ units of service, check completion first. If it is still unfinished, trigger immediately. If the first violation occurs at position $q$ with job $e=\pi(q)$, then
\[
\tau=\sum_{a<q}p_{\pi(a)}+u_e.
\]
If no upper bound is violated, set $\tau=\infty$. At a trigger the current job continues, so entering the fallback does not itself create a charged interruption.
\end{definition}

\begin{definition}[Geometric fallback]
Suppose $D>1$ and choose integer $K\ge1$. Put
\[
\begin{aligned}
c&=D^{1/K},\\
T_h&=s_0c^{h+1}\quad(h=0,\ldots,K-1),\qquad T_{-1}=0.
\end{aligned}
\]
At trigger time, place the triggering job first and then keep the inherited upper-endpoint order for the other unfinished jobs. In round $h$, serve each unfinished job until completion or until its cumulative fallback service reaches $T_h$. A job that completes at a threshold is removed before any later activation. For $D=1$ only $K=1$ is used.
\end{definition}

\paragraph{Worked example.}
Let $s_0=1$, $s_1=4$, so $D=4$, and take $K=2$. Then $c=2$ and the cumulative thresholds are $T_0=2$ and $T_1=4$. Suppose the upper-endpoint order is jobs $1,2,3$, job 1 has reported upper endpoint $u_1=1$, and the actual processing times are $p_1=3.5$, $p_2=1.5$, $p_3=1$. Job 1 is still unfinished after one unit of service, so the trigger occurs at time $\tau=1$ with residual vector $(2.5,1.5,1)$, and job 1 remains first. In round 0, job 1 receives two units of fallback service and remains unfinished, so switching from job 1 to job 2 incurs one interruption cost. Jobs 2 and 3 then complete within the same round, and departures at completion are free. In round 1, job 1 is the only survivor and its remaining $0.5$ units complete it. Thus the example has one trigger, two geometric rounds, and exactly one charged interruption.

\paragraph{Causality and finiteness.}
The trigger depends only on the service and completion history observed up to that time. For every fixed finite $K$, the fallback has finitely many service phases and completes all residual jobs because its final cumulative threshold is $s_1$. Appendix~\ref{app:meas} gives the measurability details needed for the random-instance result in Section~8.

\section{Fallback bound with interruption costs}
\begin{theorem}[Fixed-depth fallback bound]\label{thm:residual}
Let $r$ be the residual vector at fallback entry and $m=|r|$. For $D>1$ define
\[
\alpha_m(c)=1+c-\frac{2c}{m+1},
\]
\[
\beta_1(D)=0,\qquad \beta_K(D)=\max_{1\le q<K}qD^{-q/K}\quad(K\ge2).
\]
For every finite $K$ and every $\lambda\ge0$,
\[
F^{\rm res}_{G_K}(r)\le\alpha_m(D^{1/K})\operatorname{OPT}(r),
\]
\[
P^{\rm res}_{G_K}(r)\le\frac{\beta_K(D)}{s_0}\sum_i r_i,
\]
and
\[
F^{\rm res}_{G_K}(r)+\kappa P^{\rm res}_{G_K}(r)\le\Psi_{m,D,\lambda}(K)\operatorname{OPT}(r),
\]
where
\[
\Psi_{m,D,\lambda}(K)=\alpha_m(D^{1/K})+\lambda\beta_K(D).
\]
For $D=1,K=1$, residual service is SPT and the exact coefficient is one for every price.
\end{theorem}
\begin{proof}
Assume $D>1$ and write $c=D^{1/K}$. Put the triggering job first in residual order and let $h_j$ be job $j$'s completion round. Set
\[
W=\sum_i r_i,\qquad Q=\sum_{i<j}\min\{r_i,r_j\}.
\]
For $i<j$, let $A_{ij}$ be service to $i$ while $j$ is unfinished and $B_{ij}$ service to $j$ while $i$ is unfinished. Then
\[
F^{\rm res}_{G_K}(r)=W+\sum_{i<j}(A_{ij}+B_{ij}).
\]
Here $W$ accounts for the service contributing to each job's own completion time, while $A_{ij}+B_{ij}$ is exactly the additional service accumulated while the other member of the pair is still unfinished. Because $i$ precedes $j$, $A_{ij}\le\min\{r_i,T_{h_j}\}$. If $h_j\ge1$, completion-round minimality gives $r_j>T_{h_j-1}$, so $T_{h_j}<cr_j$. If $h_j=0$ and $i<j$, then $j$ is not the triggering job, hence $r_j\ge s_0$ and $T_0=s_0c\le cr_j$. Thus $A_{ij}\le c\min\{r_i,r_j\}$. Conversely, $j$ receives no service in $i$'s completion round before $i$ completes, so
\[
B_{ij}\le\min\{r_j,T_{h_i-1}\}\le\min\{r_i,r_j\}.
\]
Therefore $F^{\rm res}\le W+(c+1)Q$. The clairvoyant identity gives $\operatorname{OPT}(r)=W+Q$. Also,
\[
Q\le\frac12\sum_{i<j}(r_i+r_j)=\frac{m-1}{2}W,
\]
because each residual size appears in exactly $m-1$ unordered pairs. Substitution yields the $\alpha_m(c)$ bound. The only residual that may be below $s_0$ is the triggering job's residual, and placing that job first removes any need for a lower bound on it.

For interruptions let $L_i$ be job $i$'s number of nominal fallback activations. The final activation ends in completion and is free; each earlier activation creates at most one charged interruption. Hence $P^{\rm res}\le\sum_i(L_i-1)$. If $L_i=q+1$, strict survival through threshold $s_0D^{q/K}$ gives
\[
q<qD^{-q/K}\frac{r_i}{s_0}\le\beta_K(D)\frac{r_i}{s_0}.
\]
Summing proves the interruption bound, and $\sum_i r_i\le\operatorname{OPT}(r)$ gives the joint bound. The $D=1$ branch is immediate.
\end{proof}

\section{Choosing the fallback depth}
\begin{theorem}[Optimal depth and comparative statics]\label{thm:adaptive}
For $D>1$, $m\ge2$, and $\lambda>0$, define
\[
\begin{aligned}
K_m^*(D,\lambda)&\in\arg\min_{K\ge1}\Psi_{m,D,\lambda}(K),\\
\Psi^*_{m,D,\lambda}&=\Psi_{m,D,\lambda}(K_m^*).
\end{aligned}
\]
A finite minimizer exists. If $0<\lambda_1<\lambda_2$, every minimizer at $\lambda_2$ is no larger than every minimizer at $\lambda_1$; if $m_1<m_2$, every minimizer for $m_2$ is no smaller than every minimizer for $m_1$. Thus
\[
\lambda\uparrow\Longrightarrow K^*\downarrow,\qquad m\uparrow\Longrightarrow K^*\uparrow.
\]
Moreover $\Psi^*_{m,D,\lambda}$ is strictly increasing in $m$ and strictly improves the optimized coarse envelope $\min_K[1+D^{1/K}+\lambda\beta_K(D)]$. For $m=1$ and positive price, $K_1^*=1$; for $D=1$, $K=1$ is exact for every price. At $D>1,\lambda=0,m\ge2$,
\[
\inf_{K\ge1}\Psi_{m,D,0}(K)=\frac{2m}{m+1}
\]
is not attained at finite $K$.
\end{theorem}
\begin{proof}
Write $a_m=(m-1)/(m+1)$, so $\Psi_m(K)=1+a_mD^{1/K}+\lambda\beta_K(D)$. The first nonconstant term decreases in $K$ while $\beta_K(D)$ strictly increases and diverges, so positive price gives a finite minimizer. Adding optimality inequalities at two prices gives $(\lambda_2-\lambda_1)(\beta_{K_2}-\beta_{K_1})\le0$, hence $K_2\le K_1$. At two residual counts, adding the corresponding inequalities gives $(a_2-a_1)(D^{1/K_2}-D^{1/K_1})\le0$, hence $K_2\ge K_1$.

For strict optimized-coefficient monotonicity, let $m_1<m_2$ and let $K_2^*$ minimize the $m_2$ objective. For every finite $K$, $a_{m_2}>a_{m_1}$ and $D^{1/K}>1$, so $\Psi_{m_2}(K)>\Psi_{m_1}(K)$. Therefore
\[
\Psi^*_{m_2}=\Psi_{m_2}(K_2^*)>\Psi_{m_1}(K_2^*)\ge\Psi^*_{m_1}.
\]
At zero price the objective decreases to the displayed nonattained limit.

The positive-price optimizer is also computable by finite forward search. After evaluating successive depths, stop at the first $K\ge2$ for which
\[
1+\lambda\beta_K(D)>\min_{1\le j<K}\Psi_{m,D,\lambda}(j).
\]
Because $\beta_K$ increases and every later objective is at least $1+\lambda\beta_L(D)$, no later depth can beat the incumbent.
\end{proof}

\begin{definition}[Adaptive fallback policy $G^{\rm ad}$]
If $D=1$, use $K=1$. If $D>1$ and $\lambda>0$, then when a trigger leaves $m$ unfinished jobs, observe $m$, choose a deterministic minimizer $K_m^*(D,\lambda)$, and run the corresponding fallback. If $D>1$ and $\lambda=0$, then a residual set with $m=1$ already has exact coefficient one for every finite depth, so we may take $K=1$; for $m\ge2$ no finite bound-minimizing depth exists. We therefore do not define an exact optimized zero-price policy over the full range of possible residual set sizes. Theorem~\ref{thm:residual} remains valid for every fixed finite $K$.
\end{definition}

\paragraph{Zero-price approximation.}
The nonattainment at $\lambda=0$ does not prevent finite-depth approximation. For $D>1$ and $m\ge2$,
\[
\Psi_{m,D,0}(K)-\frac{2m}{m+1}=\frac{m-1}{m+1}\bigl(D^{1/K}-1\bigr).
\]
Hence, for any $\varepsilon>0$, every integer
\[
K\ge \max\left\{1,\left\lceil\frac{\log D}{\log\left(1+\varepsilon\frac{m+1}{m-1}\right)}\right\rceil\right\}
\]
gives $\Psi_{m,D,0}(K)\le 2m/(m+1)+\varepsilon$. Thus zero price has no finite exact optimizer for $m\ge2$, but any prescribed additive accuracy in the proved coefficient is attained at finite depth.

\section{Numerical illustration}
\begin{table*}[t]
\centering
\caption{Positive-price depth selected by the exact residual bound for $D=50$. These are deterministic evaluations of the theorem, not empirical performance measurements.}\label{tab:depth}
\begin{tabular}{rrrr}
\toprule
$m$ & $\lambda$ & $K_m^*$ & $\Psi_m^*$\\
\midrule
2&0.01&39&1.4052\\
2&0.10&14&1.5716\\
2&0.50&7&1.9099\\
2&1.00&6&2.1827\\
5&0.01&54&1.7675\\
5&0.10&18&1.9972\\
5&0.50&9&2.4489\\
5&1.00&7&2.8198\\
10&0.01&60&1.9297\\
10&0.10&21&2.1827\\
10&0.50&10&2.6738\\
10&1.00&7&3.0848\\
10&1.50&6&3.3847\\
\bottomrule
\end{tabular}
\end{table*}

Table~\ref{tab:depth} reports the bound-minimizing depths for $D=50$ across representative residual counts and interruption penalties. For $n=10$ and $D=50$, Proposition~\ref{prop:range} gives
\[
R_{10,50}=\frac{545}{104}\approx5.24038.
\]
With $\rho=2$, we have $A_{10,50,2}=2$. In the single-failure case with $m=10$ unfinished jobs, the fallback bound is smaller than the bounded-range list bound precisely when
\[
2+\Psi^*_{10,50,\lambda}<\frac{545}{104}.
\]
Equality occurs numerically at $\lambda\approx1.2379013$, with $K=7$. Thus the proved fallback bound is smaller below this value and the list bound is smaller above it.

\section{From the fallback bound to the full policy}
\begin{theorem}[Full-policy bounds]\label{thm:prop}
Assume either $D=1$ or $D>1$ with $\lambda>0$, so the adaptive policy applies, and write $A=A_{n,D,\rho}$ and $R=R_{n,D}$. Under joint coverage,
\[
J_{G^{\rm ad}}(p)\le A\operatorname{OPT}(p).
\]
If exactly the triggering job violates its interval while all other intervals cover, with realized residual count $m$,
\[
J_{G^{\rm ad}}(p)\le[A+\Psi^*_{m,D,\lambda}]\operatorname{OPT}(p).
\]
For arbitrary reports with a trigger leaving $m$ residual jobs,
\[
J_{G^{\rm ad}}(p)\le[R+\Psi^*_{m,D,\lambda}]\operatorname{OPT}(p).
\]
If no trigger occurs, $J\le R\operatorname{OPT}$.
\end{theorem}
\begin{proof}
Coverage gives the $A$ bound by Proposition~\ref{prop:range}. In the single-failure case, cap the triggering processing time at its upper endpoint, producing a covered vector $p'\le p$. Let $B$ be the pre-trigger absolute completion contribution: completion times of already completed jobs plus $m\tau$ for the residual jobs. The actual pre-trigger trajectory is a prefix of the covered upper-endpoint list on $p'$. Continuing that list after the trigger can only add nonnegative completion-time contribution, hence
\[
B\le F_U(p')\le\rho\operatorname{OPT}(p')\le\rho\operatorname{OPT}(p),
\]
where the last inequality uses $p'\le p$ and monotonicity of SPT optimum.

Independently, $B$ is no larger than the total completion-time cost obtained by continuing the same upper-endpoint list on the actual processing vector $p$. Proposition~\ref{prop:range} therefore gives
\[
B\le F_U(p)\le R\operatorname{OPT}(p).
\]
Combining the two estimates yields $B\le A\operatorname{OPT}(p)$. The zero-padded residual vector satisfies $\operatorname{OPT}(r)\le\operatorname{OPT}(p)$, and there is no charged interruption before the trigger or when entering the fallback. Adding the selected residual bound proves the single-failure result.

For arbitrary reports, only the bounded-range list argument is available, giving $B\le F_U(p)\le R\operatorname{OPT}(p)$ and therefore $R+\Psi_m^*$. With no trigger, the upper-endpoint list itself has coefficient $R$.
\end{proof}

\begin{theorem}[Expected bound for random instances]
Assume the same adaptive domain and $0<\mathbb E[\operatorname{OPT}(P)]<\infty$. Let $M$ be residual count on triggering outcomes and $M=0$ otherwise, with $\Psi_0^*=0$. Define the OPT-weighted probabilities
\[
\begin{aligned}
q_{1,m}&=\frac{\mathbb E[\operatorname{OPT}\mathbf1_{V_1\cap\{M=m\}}]}{\mathbb E[\operatorname{OPT}]},\\
q_{c,m}&=\frac{\mathbb E[\operatorname{OPT}\mathbf1_{V_c\cap\{M=m\}}]}{\mathbb E[\operatorname{OPT}]}.
\end{aligned}
\]
Then, with $q_1=\sum_{m=1}^nq_{1,m}$ and $q_c=\sum_{m=0}^nq_{c,m}$,
\[
\begin{aligned}
\frac{\mathbb E[J_{G^{\rm ad}}]}{\mathbb E[\operatorname{OPT}]}
&\le (1-q_1-q_c)A\\
&\quad+\sum_{m=1}^nq_{1,m}(A+\Psi_m^*)\\
&\quad+\sum_{m=0}^nq_{c,m}(R+\Psi_m^*).
\end{aligned}
\]
For fixed finite $n$ and $s_0\le P_j\le s_1$, the integrability assumption is automatic.
\end{theorem}

\section{Comparison with continuing the initial list}
\begin{proposition}[When the fallback bound improves on the list bound]
Continuing the upper-endpoint list has no charged interruptions and is bounded by $R_{n,D}$. In the single-failure case with residual count $m$, Theorem~\ref{thm:prop} gives the fallback coefficient $A_{n,D,\rho}+\Psi^*_{m,D,\lambda}$. Hence the fallback bound is smaller than the bounded-range list bound exactly when
\[
A_{n,D,\rho}+\Psi^*_{m,D,\lambda}<R_{n,D}.
\]
If $\rho\ge R_{n,D}$ the inequality cannot hold under these bounds; if $\rho<R_{n,D}$ it is equivalent to $\Psi_m^*<R_{n,D}-\rho$. Under arbitrary reports the available fallback coefficient is $R_{n,D}+\Psi_m^*$, so it cannot improve on $R_{n,D}$ within this analysis. This proposition compares only the bounds proved here; it does not claim that $R_{n,D}$ is the best continuation bound available after observing a particular failure state.
\end{proposition}

\section{Related work}
Classical nonclairvoyant scheduling already treats preemption and its cost as important modelling issues. \citet{motwani1994} develop a general framework for preemptive nonclairvoyant scheduling whose stated model explicitly takes preemption cost into account. Our setting is more specific: we use the additive objective $J=F+\kappa P$, with one charge for each switch away from an unfinished job, and change scheduling rule only after violation of a reported upper bound. \citet{bartal2006} study a constant cost per preemption together with flow/completion performance; unlike our all-release interval setting, their jobs have known sizes upon release and arrive online.

Restart models include \citet{vanstee2005}, \citet{amouzandeh2026}, and \citet{jager2025}. J\"ager et al. are particularly close structurally: their $b$-scaling strategies probe unfinished jobs at geometric service scales. Their main construction uses kill-and-restart/nonclairvoyant semantics, whereas our fallback begins only after a reported upper bound is violated, preserves completed service, and charges interruptions explicitly. \citet{liu2004} instead model preemption penalties as setup time that consumes machine capacity.

An especially close line is scheduling with testing, or explorable uncertainty. \citet{durr2020} study an all-release single-machine total-completion-time problem in which each job has a known upper limit and the scheduler may spend machine time on a test that reveals the exact processing time. \citet{albers2021} extend the framework to non-uniform testing times and a preemptive variant, while \citet{gong2025} continue the testing line on identical parallel machines in the \emph{Journal of Scheduling}. Our information structure has no separate test and never reveals the exact processing time. New information arrives only when ordinary execution exceeds a reported upper endpoint; completed service is retained, and later interruptions are charged through $\kappa P$. The decision here is therefore how to reschedule after a report has failed, rather than whether to spend time acquiring exact information.

Interval processing times are well established in robust scheduling. \citet{sotskov2018} study single-machine total completion time with interval processing times using stability and optimality-box ideas; \citet{kasperski2008}, \citet{pereira2016}, and \citet{bold2022} study related minmax-regret or recoverable-robust models. These approaches optimize over uncertainty sets or scenarios rather than reacting to an upper-bound violation during execution. \citet{tao2010} instead exploit a bounded ratio of longest to shortest processing time, closely related to our parameter $D=s_1/s_0$, but without interval reports, a violation trigger, or interruption pricing.

Learning-augmented and partially revealed scheduling use other information channels. \citet{purohit2018}, \citet{lindermayr2022}, \citet{benomar2024}, \citet{benomar2025}, and \citet{lindermayr2026} study point, permutation, partial, or progressively revealed predictions. \citet{zhao2026} give smoothness bounds for non-clairvoyant scheduling with job-size predictions.

A particularly close recent preprint is Blue, Im, and Lindermayr, \emph{Learning-Augmented Online Scheduling with Parsimonious Preemption} (arXiv:2605.23255). They minimize total completion time using job-size predictions while simultaneously controlling the number of preemptions, including constant-competitive single-machine guarantees whose preemption count depends on prediction error. Their preemption count is a separate complexity/overhead measure rather than an additive term in the optimized objective. Our model instead uses interval reports, begins with a nonpreemptive upper-endpoint list, switches to the geometric fallback only after an observable upper-bound violation, and prices interruptions directly in $J=F+\kappa P$. Thus both papers address prediction quality and preemption, but the objective and information pattern are different.

Two other recent preprints are adjacent. Chen, Ye, and Zhou, \emph{Adaptively Robust LLM Inference Optimization under Prediction Uncertainty} (arXiv:2508.14544), use output-length prediction intervals with adaptive refinement for LLM inference. Gupta, Kaplan, Lindermayr, Schl\"oter, and Yingchareonthawornchai, \emph{Better Late Than Never: Online Flow Time Scheduling with Online Estimates} (arXiv:2609.07402), study online flow-time scheduling when approximate processing-time estimates may become available only after processing has begun; they have online arrivals and delayed approximate-size information rather than interval-bound violations and a charged fallback. \citet{buld2025} study testing for total weighted completion time, while \citet{gupta2026} study robust flow-time scheduling under multiplicative estimate error.

Taken together, the closest lines already contain costly preemption, geometric probing, interval uncertainty, testing, and learning-augmented scheduling in different forms. The contribution here is the specific all-release model in which ordinary execution itself exposes a failed upper report, completed service is retained, interruptions are priced additively, the fallback depth is optimized from the observed residual count, and the policy is compared against the exact bounded-range list coefficient $R_{n,D}$.

\section{Limitations and conclusion}
The model is deliberately limited to a finite all-release batch on one unit-speed resumable machine, with a known finite processing-time range and an additive cost for switching away from unfinished work. It excludes release times, weights, due dates, setup durations, multiple machines, routing, dynamic report updates, and nonstationary calibration. Appendix~\ref{app:meas} records a small minimax property of the geometric thresholds within the stated common-threshold family; no unrestricted threshold-minimax result is claimed.

Within this setting, a violated upper bound provides an observable point at which the scheduler can change its treatment of the unfinished jobs. The coefficient $R_{n,D}$ is the sharp worst-case bounded-range coefficient for an arbitrary nonpreemptive list and supplies the comparison used in Proposition~9.1. Every fixed finite $K$ is valid at any nonnegative interruption penalty; when the penalty is positive, the observed number of unfinished jobs determines an attained minimizing depth. At $D=1$, $K=1$ is exact. At $D>1$ and zero price with $m\ge2$, the infimum over depths is not attained, although Section~6 gives an explicit finite depth for any prescribed additive accuracy in the bound. We do not claim global optimality over all resumable policies, a matching lower bound for all policies, or optimality of $R_{n,D}$ among continuation rules that use the full state observed at a failure.

\appendix
\section{Fractional comparison and static permutations}\label{app:fractional}
\begin{proposition}[Fractional comparison]
Let $0<\theta<1$. Policy $H_\theta$ gives rate $1-\theta$ to the first unfinished job in upper-endpoint order and shares rate $\theta$ equally over all unfinished jobs. On joint coverage,
\[
F_{H_\theta}(p)\le\frac{\rho}{1-\theta}\operatorname{OPT}(p).
\]
For every positive processing vector,
\[
F_{H_\theta}(p)\le\frac{2n}{(n+1)\theta}\operatorname{OPT}(p),
\]
and the underlying equal-sharing factor $2n/(n+1)$ is tight at equal job sizes. The corresponding OPT-weighted expected bound follows by conditioning on joint coverage.
\end{proposition}
\begin{proof}
For a fixed list $\pi$,
\[
F_\pi(p)=\sum_{k=1}^n(n-k+1)p_{\pi(k)}.
\]
On coverage, $p\le u\le\rho p$, so upper-endpoint order costs at most $\rho\operatorname{OPT}(p)$. For equal sharing let $p_{(1)}\le\cdots\le p_{(n)}$, $p_{(0)}=0$, and $d_k=p_{(k)}-p_{(k-1)}$. The exact equal-sharing cost is
\[
F_E(p)=\sum_{i=1}^n[2(n-i)+1]p_{(i)}.
\]
For suffix size $m=n-k+1$, the coefficient of $d_k$ is $m^2$, while its coefficient in OPT is $m(m+1)/2$. Hence
\[
F_E(p)\le\frac{2n}{n+1}\operatorname{OPT}(p),
\]
with equality for equal positive job sizes.

Whenever the next unfinished job in upper-endpoint order is present, $H_\theta$ gives it rate at least $1-\theta$ from the priority component, in addition to any sharing service. Inducting over completions gives $F_{H_\theta}\le F_U/(1-\theta)$. When $m$ jobs are unfinished, the sharing component gives every unfinished job rate $\theta/m$, exactly the per-job rate of equal sharing at total speed $\theta$; the extra priority service can only complete jobs earlier. Induction over completion epochs therefore gives $F_{H_\theta}\le F_E/\theta$. Combining the two bounds proves the proposition, and OPT-weighted conditioning gives the expected version.
\end{proof}

\begin{proposition}[Static-permutation obstruction]
For identical reports, a randomized static policy has some job with expected list coefficient at least $(n+1)/2$. Making that job's processing time tend to infinity while the others stay at one gives asymptotic expected competitive ratio at least $(n+1)/2$.
\end{proposition}
\begin{proof}
For any realized permutation, the list coefficients assigned to the $n$ job labels are exactly $n,n-1,\ldots,1$, so their average is $(n+1)/2$. Averaging also over the policy's randomization, at least one fixed job $j^*$ therefore has expected coefficient at least $(n+1)/2$. Set every other processing time to one and let $p_{j^*}=M\to\infty$. The policy's expected total-completion-time cost is $M\,\mathbb E[c_{j^*}]+O(1)$, where $c_{j^*}$ is the list coefficient of $j^*$, while SPT puts the long job last and has $\operatorname{OPT}=M+O(1)$. Dividing and sending $M\to\infty$ gives an asymptotic expected competitive ratio at least $(n+1)/2$.
\end{proof}

\section{Measurability and threshold details}\label{app:meas}
\paragraph{Stopping-time and Borel details.}
With upper-endpoint order $\pi$, define
\[
E_q=\bigcap_{a<q}\{p_{\pi(a)}\le u_{\pi(a)}\}\cap\{p_{\pi(q)}>u_{\pi(q)}\}.
\]
These finite-coordinate events are Borel and disjoint. On $E_q$, $\tau_q=\sum_{a<q}p_{\pi(a)}+u_{\pi(q)}$; set $\tau=\infty$ on the no-trigger event. Hence $\{\tau\le t\}$ is a finite union of observable events. Fixed-$K$ sorting, threshold comparisons, finite service phases, and arithmetic are Borel constructions, so completion times, residual vector, $P$, $J$, and $M$ are measurable. In the adaptive domain, $K_m^*$ is a deterministic function of the observed finite-valued $m$.

\paragraph{Monotonicity and divergence of $\beta_K$.}
For $D>1$, $\beta_1=0$. If $q$ maximizes $qD^{-q/K}$, the same $q$ is admissible at $K+1$ and yields a strictly larger value, so $\beta_{K+1}>\beta_K$. Choosing $q=\lfloor K/2\rfloor$ gives $\beta_K(D)\ge\lfloor K/2\rfloor D^{-1/2}\to\infty$.

\paragraph{Remark B.3 (a minimax property of the geometric thresholds).}
Consider a common sequence of cumulative thresholds $0<T_0<\cdots<T_{K-1}$ with $T_{K-1}\ge s_1$, including the initial factor $T_0/s_0$. Multiplying the $K$ adjacent expansion factors from $s_0$ through $T_{K-1}$ shows that their largest value is at least $D^{1/K}$. The geometric ladder attains equality. This statement is only about the stated common-threshold family.

\section*{Statements and Declarations}
\paragraph{Funding.} No external funding was received for this work.

\paragraph{Competing interests.} The author declares no financial or non-financial competing interests.

\paragraph{Data availability.} No empirical dataset was used. All numerical values reported in the manuscript are deterministic calculations from formulas stated in the paper.

\paragraph{Author contribution.} Lachlan Bridges: conceptualization, methodology, formal analysis, investigation, writing---original draft, writing---review and editing.

\paragraph{Use of generative AI.} OpenAI ChatGPT was used for mathematical exploration, proof checking, literature searches, and language editing. The author checked all mathematical claims, calculations, citations, and final text and takes responsibility for the manuscript.

\end{document}